\documentclass[preprint,pre,floats,aps,amsmath,amssymb,12]{article}
\usepackage[T1]{fontenc} 
\usepackage[english]{babel} 
\usepackage[hmarginratio=1:1,top=32mm,columnsep=20pt]{geometry} 
\usepackage{graphicx}
\usepackage{bm}
\usepackage{abstract} 

\usepackage{titlesec} 
\titleformat{\section}[block]{\large\scshape\centering}{\thesection.}{1em}{} 
\titleformat{\subsection}[block]{\large}{\thesubsection.}{1em}{} 
\usepackage{titling} 

\usepackage{bookman} 
\usepackage{hyperref} 
\usepackage{amssymb}
\usepackage{amsmath}
\usepackage{adjustbox}
\usepackage{graphicx}
\usepackage{caption}
\usepackage{subfig}
\usepackage{float}
\usepackage{fancyhdr} 
\usepackage{amsthm}
\newtheorem{theo}{Theorem}[section]

\newtheorem{rem}{Remark}[section]
\providecommand{\keywords}[1]
{
	\small	
	\textbf{\textit{Keywords---}} #1
}
\title{Impact studies of nationwide measures COVID-19 anti-pandemic: compartmental model and machine learning}
\author{\thanks{Support of the Non Linear Analysis, Geometry and Applications (NLAGA) Project}
	Mouhamadou A.M.T. Bald\'e$^a$, Coura Bald\'e$^b$ and Babacar M. Ndiaye$^a$\\
	University of Cheikh Anta Diop.  \\ 
	BP 45087, 10700. Dakar, Senegal.\\
	$^a$ Laboratory of Mathematics of Decision and Numerical Analysis (LMDAN).\\	
	Department of Mathematics of Decision(DMD)-FASEG.\\
	\href{mailto:mouhamadouamt.balde@ucad.edu.sn}{mouhamadouamt.balde@ucad.edu.sn}\\
	\href{mailto:babacarm.ndiaye@ucad.edu.sn}{babacarm.ndiaye@ucad.edu.sn}\\
	$^b$ Laboratory of Applied Mathematics (LMA)-FST\\
	\href{mailto:coura.balde@ucad.edu.sn}{coura.balde@ucad.edu.sn}\\
}	
\date{}
\renewcommand{\maketitlehookd}{%
	\begin{abstract}
		\noindent In this paper, we deal with the study of the impact of nationwide measures COVID-19 anti-pandemic. We drive two processes to analyze COVID-19 data considering measures. We associate level of nationwide measure with value of parameters related to the contact rate of the model. Then a parametric solve, with respect to those parameters of measures, shows different possibilities of the evolution of the pandemic. Two machine learning tools are used to forecast the evolution of the pandemic. Finally, we show comparison between deterministic and two machine learning tools.   
	\end{abstract}
	\keywords{COVID-19, Nationwide measures analysis, fitting, forecasting, machine learning, asymptomatic,  symptomatic}
}
\begin{document}

\maketitle
	
\section{Introduction}
\label{sec:intro}
The COVID-19 pandemic is testing the entire world so that measures are being taken in most nations to stem its development. These measures can generally be of different kinds such as social distancing, partial or total confinement, etc. It would therefore be interesting to be able to effectively analyze the effects of the measures taken on the spread of the pandemic. Here we offer an analysis of the impact of the measures taken that we apply to the case of Senegal. In previous papers \cite{Balde:2020}, \cite{Ndiayeetal:2020} and \cite{Ndiayeetal2:2020}  we proposed a start of study using the results of \cite{LMSW:2020}.\\ 
In \cite{LMSW:2020}, a useful method for the study of the evolution of COVID-19 pandemic has been resented by using a compartmental model with Susceptible, Infected asymptomatic, Infected reported symptomatic and unreported symptomatic (SIRU). In \cite{Balde:2020}, the author use that method to study the COVID-19 spread in Senegal with a classical SIR model.\\
In this work we aim to analyze the impact of the anti pandemic measures taken in Senegal. It is a continuation of the work done in \cite{Balde:2020}. We do a two-step analysis. The first uses the differential equations model denoted SIRU introduced in \cite{LMSW:2020}, and the second uses two machine learning tools: Predict of Wolfram Mathematica using Neural Networks method and Prophet.  \\
We conduct the work in the following way. In the section \ref{sec:analys} we study the effect of the nationwide measures by using an epidemic model presented in \cite{LMSW:2020} and we present some machine learning tools. We show, in the section \ref{sec:simul}, the numerical results. In the section \ref{sec:discuss}, we discuss the results. Then in section \ref{sec:app}, we perform an analysis of the model and explain the parameters estimation. Finally in the section \ref{sec:conclusion}, we end by making a conclusion and advancing perspectives.
\section{Analysis}
\label{sec:analys}
\subsection{Analysis of the measures}
\label{subsec:measures}
In \cite{Balde:2020} a classical SIR model was studied using results from the paper \cite{LMSW:2020}. The aim was to analyze the effect of the nationwide measures using the data after nationwide measures. In fact, throughout $ t_{0} $ to $ T $, we fitted an exponential function to the data of the total cases of infection of this period. $ T $ represents the date of the nationwide measures, and $ t_{0} $ is the starting time of the epidemic. We consider that the effects of the measures are such that they lead to a reduction in the contact rate. To describe this reduction, we chose a slowly decreasing function over time. We consider that the measures taken are not strong enough to systematically drop the contact rate to $  0  $. This new function corresponds to the first one and the data in the period before measures from $ t_{0} $ to $ T $, then takes a slower trajectory than that of the first function after $ T $. In other words, from the date $ T $, the new curve goes under the old one.
We consider that if on the dates $ t> T $ the data goes under the new curve obtained with a contact rate after measures then, we can say that these measures affect the evolution of the pandemic.\\
In the previous paper \cite{Balde:2020}, the function we fitted to the data from $2020$ March $02$ to March $31$ by least square method, is $TNI(t)=b\exp(c t)-a$ with $a=13.9324,\ b=9.61779$ and $c=0.100095$ (figure \ref{fig:fitfunc1old}). In this paper we fit a new exponential function to the data from $2020$ March $02$ to April $25$ (figure \ref{fig:fitfunc1}). For this new function, we have $a = 99.9214,\ b = 81.325,\ c = 0.0371767$.
\begin{figure}[H]
	\centering
	\subfloat[Fit with data of the total cases: $TNI(t)$ is the blue line and the data are the red dotted. From $2020$ March $02$ to March $31$.\label{fig:fitfunc1old}]{\includegraphics[width=0.4\linewidth]{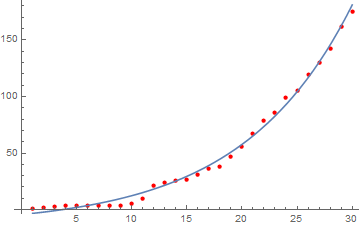}}\qquad
	\subfloat[Fit with data of the total cases: $TNI(t)$ is the blue line and the data are the red dotted. From $2020$ March $02$ to April $25$.\label{fig:fitfunc1}]{\includegraphics[width=0.4\linewidth]{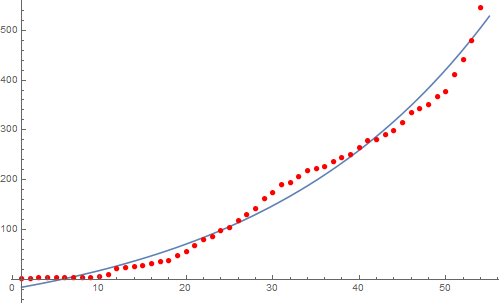}}\qquad\qquad
	\caption{Plot of the exponential curves fitting the total number of case Senegal's data.}
	\label{fig:tntsirdata}
\end{figure}	
\noindent To continue in our analysis, we will use a differential equations model introduced in \cite{LMSW:2020}. This model is as follows:
\begin{equation}
\left\{ \begin{array}{l}
\displaystyle \frac{dS}{dt}=-\beta S(t) (I(t)+I_{U}(t))\\\\

\displaystyle \frac{dI}{dt}=\beta S(t)( I(t)+I_{U}(t)) - \nu I(t)\\\\

\displaystyle \frac{dI_{R}}{dt}=\gamma\nu I(t)-\eta I_{R}(t)\\\\

\displaystyle \frac{dI_{U}}{dt}=(1-\gamma)\nu I(t)-\eta I_{U}(t)
\end{array}
\right.
\label{sir1}
\end{equation}

\noindent The initial conditions are $S(t_{0} )=S_{0}\geq 0,\ I(t_{0} )=I_{0}\geq 0$ $I_{R}(t_{0} )=I_{R0}\geq 0$, and $I_{U}(t_{0} )=I_{U0}\geq 0$ and $t_{0}$ the initial time of the epidemic. With $S(t)$ represents the susceptible individuals, $I$ the infected asymptomatic individuals, $I_{U}$ the unreported infected symptomatic and $I_{R}$ the reported infected symptomatic individuals. \\

\begin{figure}[H]
	\centering
	\includegraphics[width=0.75\linewidth]{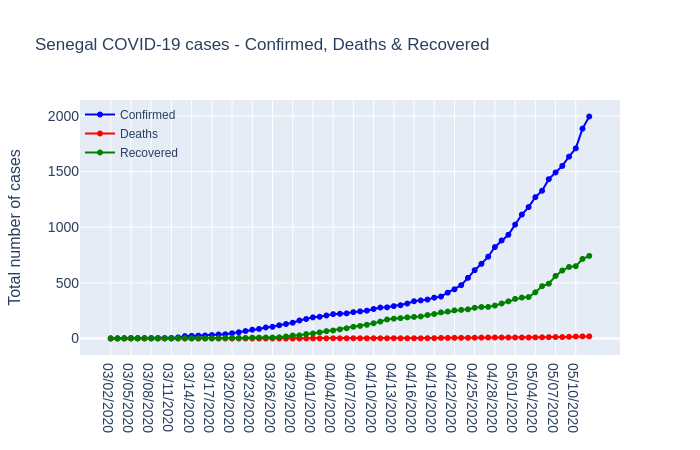}
	\caption{Plot of data of total confirmed, death and recovered from $2020$ March $02$ to May $10$.}
	\label{fig:allcasesdata}
\end{figure}	
We can also consider a version with removed compartments, the recovered $R$ and the death $D$:
\begin{equation}
\left\{ \begin{array}{l}
\displaystyle \frac{dS}{dt}=-\beta S(t) (I(t)+I_{U}(t))\\\\

\displaystyle \frac{dI}{dt}=\beta S(t)( I(t)+I_{U}(t)) - \nu I(t)\\\\

\displaystyle \frac{dI_{R}}{dt}=\gamma\nu I(t)-\eta I_{R}(t)\\\\

\displaystyle \frac{dI_{U}}{dt}=(1-\gamma)\nu I(t)-\eta I_{U}(t)\\\\

\displaystyle \frac{dR}{dt}=\alpha\eta (I_{R}(t)+I_{U}(t))\\\\

\displaystyle \frac{dD}{dt}=(1-\alpha)\eta (I_{R}(t)+I_{U}(t))
\end{array}
\right.
\label{sir2}
\end{equation} 
\noindent Let's present the parameters. $\beta$ is the contact rate. $1/\nu$ is the average time during which asymptomatic infectious are asymptomatic. $\gamma$ is the fraction of asymptomatic infectious individuals that become reported symptomatic infectious. $1/\eta$ is the average time symptomatic infectious have symptoms. $\gamma\nu$ is the rate at which asymptomatic infectious become reported symptomatic. $(1-\gamma)\nu$ is the rate at which asymptomatic infectious become symptomatic but unreported. $\alpha$ is the proportion of recovered and $1-\alpha$ is the proportion of death due to the infection.\\

\noindent We make the assumption that the function we fit the total number of reported infected cases is given by $\displaystyle \gamma\nu\int_{t0}^{t} I(s) ds$. \\
The starting time of the pandemic, the initial conditions and some parameters of the model are estimated and others are fixed. The parameters set are $\gamma$, $1/\nu$ and $1/\eta$. The values chosen for $1/\nu$ and $1/\eta$ are those used by the medical authorities i.e. $1/7$. Regarding the proportion of unreported and reported, we also use the same value proposed in \cite{LMSW:2020}, ie $80\%$ are reported and $20\%$ are unreported. Further detailed studies could make it possible to estimate the good value of this proportion for Senegal where we apply the model. For the estimation of the  parameters, the starting time and the initial conditions, we use the calculation methods detailed in \cite{LMSW:2020} and \cite{Balde:2020}. We start with the fitting exponential function which we equalize with the integral formula above $\displaystyle TNI(t)=b\exp(c t)-a=\displaystyle \gamma\nu\int_{t0}^{t} I(s) ds$. Hence some calculation of derivations, replacement in the model and integration gives:
\begin{itemize}
	\item  $\displaystyle t_{0}=\frac{ln(a)-ln(b)}{c}$.     
	\item $\displaystyle I(t)=I(t_{0})\exp(c(t-t_{0}))$, with $\displaystyle I(t_{0})=I_{0}=\frac{a c}{\gamma\nu}$.
	\item $\displaystyle I_{U}(t)=I_{U}(t_{0})\exp(c(t-t_{0}))$, with $\displaystyle I_{U}(t_{0})=I_{U0}=\frac{(1-\gamma)a c}{\gamma(\eta+c)}=\frac{(1-\gamma)\nu}{\eta+c}I_{0}$.
	\item $\displaystyle I_{R}(t)=I_{R}(t_{0})\exp(c(t-t_{0}))$, with $\displaystyle I_{R}(t_{0})=I_{0}-I_{U0}=\frac{a c}{\eta+c}=\frac{\gamma\nu}{\eta+c}I_{0}$.
	\item $\displaystyle R(t)=\frac{\alpha\eta\nu(I(t)-I_{0})}{c(c+\eta)}$.
	\item $\displaystyle D(t)=\frac{(1-\alpha)\eta\nu(I(t)-I_{0})}{c(c+\eta)}$.
	\item $\displaystyle \beta=\frac{c+\nu}{S_{0}}\frac{\eta+c}{(1-\gamma)\nu+\eta+c}$.
	\item $\displaystyle \mathcal{R}_{0}=\frac{c+\nu}{\nu}\frac{\eta+c}{(1-\gamma)\nu+\eta+c}(1+\frac{(1-\gamma)\nu}{\eta})$.\\
\end{itemize}      
For more details see section \ref{sec:app}.\\

\noindent We consider that after the measures taken at the time $T$, the contact rate depends on time following a formula we choose. 
We use two formulas. One of them was introduced in \cite{Balde:2020}, and the second one was proposed in \cite{LMSW2:2020}. 
The first one is :
\begin{equation}
\tilde{\beta}(t)=\left\{
\begin{array}{ll}
\displaystyle\beta & \textrm{if } t\in [t_{0},\ T]\\
\displaystyle\beta (\frac{T}{t})^{\delta/p}& \textrm{if } t>T, 
\end{array}
\right.
\label{beta1}
\end{equation}
where $\delta$ and $p$ are parameters to choose.
\noindent The second one is:
\begin{equation}
\tilde{\beta}(t)=\left\{
\begin{array}{ll}
\displaystyle\beta & \textrm{if } t\in [t_{0},\ T]\\
\displaystyle\beta \exp(-\varphi(t-T))& \textrm{if } t>T, 
\end{array}
\right.
\label{beta3}
\end{equation}
where $\varphi$ is a parameter to choose.   \\
Then the new model to  solve is:
\begin{equation}
\left\{ \begin{array}{l}
\displaystyle \frac{dS}{dt}=-\tilde{\beta} S(t) (I(t)+I_{U}(t))\\\\

\displaystyle \frac{dI}{dt}=\tilde{\beta} S(t)( I(t)+I_{U}(t)) - \nu I(t)\\\\

\displaystyle \frac{dI_{R}}{dt}=\gamma\nu I(t)-\eta I_{R}(t)\\\\

\displaystyle \frac{dI_{U}}{dt}=(1-\gamma)\nu I(t)-\eta I_{U}(t)
\end{array}
\right.
\label{sir3}
\end{equation}
By solving \eqref{sir3}, we get $\tilde{I}(t)$ from which we can calculate the total number of infected with the formula $\displaystyle \tilde{TNI}(t)=\gamma\nu\int_{t0}^{t} \tilde{I}(s) ds$. We choose values of parameters of $\tilde{\beta}$ such that $\tilde{TNI}(t)$ follows the same path that $TNI(t)$. 
Then, by way of parameters $\delta$ and $p$ in $\tilde{\beta}$, we can evaluate the measures. We do parametric solve with respect to parameters $\delta$ and $p$. Results are shown in figure \ref{fig:siru0}. 

\noindent Now we consider at time $T_{2}$, the nation applies stronger measures. We simply interpret as the contact rate $\tilde{\beta}$ is reduced by a factor $\phi\in [0,1]$. When the value of  $\phi$ is equal to $0$, it means that there are no measures, while when $\phi$ is equal to $1$, it means that the strongest measures are taken. Hence the measures are quantified with the values of proportion $\phi$. Then the contact rate become:
\begin{equation}
\tilde{\beta}_{2}(t)=\left\{
\begin{array}{ll}
\displaystyle\tilde{\beta}(t) & \textrm{if } t\in [t_{0},\ T_{2}]\\
\displaystyle(1-\phi)\tilde{\beta}(T_{2})& \textrm{if } t>T_{2}, 
\end{array}
\right.
\label{beta4}
\end{equation}         
We solve the new model with $\tilde{\beta}$ replaced by $\tilde{\beta}_{2}$. We do a parametric solve with respect to the parameter $\phi$, and we plot the result for different values of $\phi$. Then, we show different values of the peak,  depending on the values of $\phi$. Then, we can evaluate the maximal value of infection peak with respect to the level of the measures. Results are shown in figures \ref{fig:siru11}, \ref{fig:siru12}, \ref{fig:siru31} and \ref{fig:siru32}.    
\subsection{Artificial Neural Networks}
\label{ann}
Artificial neural networks are part of artificial intelligence. Biological neural networks inspire them. Biological neural networks are part of the animal brain. One of the main functions of the brain is to process information, and the primary information processing element is the neuron. This specialized brain cell combines (usually) several inputs to generate a single output. Depending on the animal, an entire brain can contain anywhere from a handful of neurons to more than a hundred billion, wired together. The output of one cell feeding the input of another, to create a neural network capable of remarkable feats of calculation and decision making (See \cite{Newman:2010}). 
If we could qualify the brain as a computer, then we would say that it is the best of computers. For this reason, the engineer seeks to improve mechanical computers to be closer to the biological computer, i.e., the brain.
The more neural connections there are, the more the network can solve complex problems. Pattern recognition is a task that neural networks can easily accomplish. For this task, introducing as input a pattern to a neural network, yields as output a pattern back (See \cite{Heaton:2015}). \\
In general, neural network problems involve a dataset used to predict values for later datasets. For that, the neural network needs to be trained. Then, neural networks can predict the outcome of entirely new datasets based on training from old data sets. \\  
Most neural network structures use some type of neuron, node, or unit. An algorithm called a neural network will generally be made up of individual interconnected neurons, see figure \ref{fig:neuron}. \\
The artificial neuron receives input from one or more sources, which may be other neurons or data entered into the network from a computer program see figure \ref{fig:neuralnet}. This entry is usually a floating-point or binary. Often the binary input is coded floating point representing true or false like $ 1 $ or $ 0 $. Sometimes the program also describes the binary input as using a bipolar system with true as $ 1 $ and false as $ -1 $.
An artificial neuron multiplies each of these inputs by a weight. It then adds these multiplications and transmits this sum to an activation function given by:
\begin{equation}
f(x_{i},w_{i})=\phi(\sum_{i=1}^{n}(w_{i}\cdot x_{i})),
\label{eq:af}
\end{equation}
with the variables $ x $ and $ w $ represent the input and the weights of the neuron, $n$ is the number of input and weight.\\
Some neural networks do not use an activation function. 
To read more about Neural Networks one can see \cite{Alpaydin:2010}, \cite{Heaton:2015}, \cite{GoodfellowBengioCourville:2016} and \cite{Newman:2010}.\\

\noindent Neural networks are implemented in machine learning tools of many software like Wolfram Mathematica\cite{WR}. We use the machine learning tool ``Predict'' to forecast the evolution of the COVID-19 pandemic in Senegal. We can choose different method of regression algorithm: ``RandomForest'', ``LinearRegression'', ``NeuralNetwork'', ``GaussianProcess'', ``NearestNeighbors'', etc.\\
We use the ``NeuralNetwork'' regression algorithm, which predicts using artificial neural networks. \\     
We consider two cases in the forecasting. The first case, only use the total number of infected case data in the training of the neural networks. While in the other forecasting, we use two types of data: the total number of infected cases and the contact rate. In the second case, the aim is to do forecasting by considering the effect of the contact rate. It is a way to show the effect of the nationwide measures taken at the time $T$, as specified in section \ref{sec:analys}. For the contact rate we use as data $\tilde{\beta}$ given by \eqref{beta1}.  
\begin{figure}[H]
	\centering
	\includegraphics[width=0.4\linewidth]{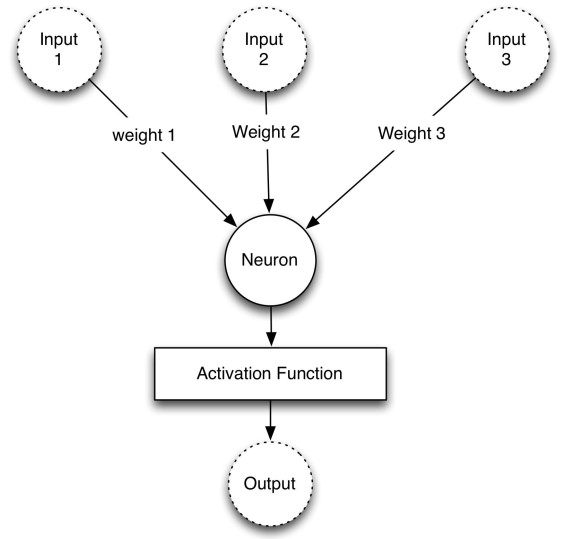}	
	\caption{Example of artificial neuron from \cite{Heaton:2015}.}
	\label{fig:neuron}
\end{figure}
\begin{figure}[H]
	\centering
	\includegraphics[width=0.4\linewidth]{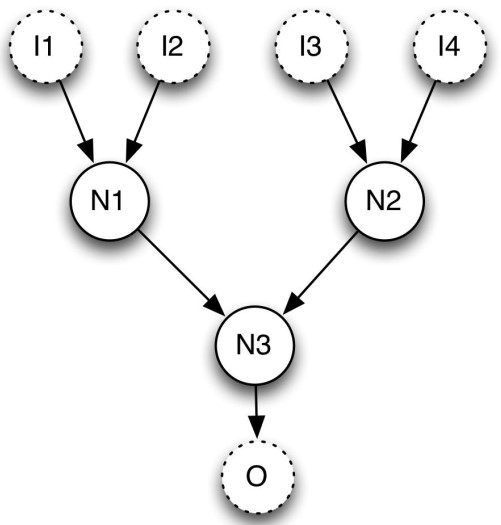}	
	\caption{Example of artificial neural networks from \cite{Heaton:2015}. I1, I2 and I3 mean input 1,2 and 3. N1, N2 and N3 mean neuron 1,2 and 3. O means output.}
	\label{fig:neuralnet}
\end{figure}

\subsection{Forecasting using Prophet}
In this section, we develop another machine learning tool for forecasting to compare with the previous SIRU models and Neural Networks method. 

\noindent We use Prophet \cite{prophet}, a procedure for forecasting time series data based on an additive model where non-linear trends are fit with yearly, weekly, and daily seasonality, plus holiday effects. \\
For the average method, the forecasts of all future values are equal to the average (or “mean”) of the historical
data. If we let the historical data be denoted by $y_1,...,y_T$, then we can write the forecasts as

$$
\hat{y}_{T+h|T}=\bar{y}=(y_1+y_2+...+y_T)/T
$$
\noindent The notation $\hat{y}_{T+h|T}$ is a short-hand for the estimate of $y_{T+h}$  based on the data $y_1,...,y_T$. 
A prediction interval gives an interval within which we expect $y_t$  to lie with a specified probability. For example, assuming that the forecast errors are normally distributed, a 95\% prediction interval for the  $h$-step forecast is 

$$
\hat{y}_{T+h|T}\pm1.96\hat{\sigma_h}
$$

\noindent where  ${\sigma_h}$ is an estimate of the standard deviation of the $h$-step  forecast distribution.  \\
For the data preparation, when we are forecasting at the country level, for small values, forecasts can become negative. To counter this, we round negative values to zero. Also, no tweaking of seasonality-related parameters and additional regressors are performed.      
\section{Numerical Simulations}
\label{sec:simul}

\subsection{Numerical analysis}
\label{sec:simul1}
The data for Senegal, we use, is obtained from daily press releases on the COVID-19 from the Ministry of Health and Social Action (\url{http://www.sante.gouv.sn/}).

\noindent The figures \ref{fig:siru11}, \ref{fig:siru12}, \ref{fig:siru31} and \ref{fig:siru32} show results related to \ref{subsec:measures}. The figures \ref{fig:siru11}, \ref{fig:siru12} correspond to $\tilde{\beta}$ \ref{beta1} and the figures \ref{fig:siru31}, \ref{fig:siru32} correspond to $\tilde{\beta}$ \ref{beta3}. \\
We use $\gamma = 0.8,\  \nu = 1/7,\ \eta = 1/7,$  the total population of Senegal is $N=16743927$ from 
Senegal Population (2020) - Worldometer (www.worldometers.info). Then we obtain: $t_{0}=5.53923,\ I_{0}=32.504, \ I_{U0}=5.1584$, $I_{R0}=27.3456$, $S_{0}=N-I_{0}=1.67439\times10^7$, $\mathcal{R}_{0}=1.30516$, $\beta=9.27954\times 10^{-9}$.  

\noindent The time of the nationwide measures in Senegal taken at 2020, March 23 is considered. Then $T=23$. For $\tilde{\beta}$ given by \eqref{beta1}, results are shown by figures \ref{fig:fitfunccomp12} and \ref{fig:fitfunccomp11}. For $\tilde{\beta}$ given by \eqref{beta3}, results are shown by figures \ref{fig:fitfunccomp2} and \ref{fig:fitfunccomp1}. We see that the maximal number of reported case goes up to $340000$ with the time peak at $t=310$, which correspond to 2021, January 05.  \\
A parametric plot of a parametric resolution of the SIRU model \eqref{sir3}, using $\tilde{\beta}$ given by \eqref{beta1}, with respect to $\delta$ and $p$, is shown in figures \ref{fig:siru0}.

\noindent Now, we consider that stronger measures is taken as explained in the subsection \ref{subsec:measures}. The results is shown in figures \ref{fig:fitfunccomp13zoom}, \ref{fig:fitfunccomp3zoom}, \ref{fig:infPredict11}, \ref{fig:infPredict1}, \ref{fig:infRUPredict11}, \ref{fig:infRUPredict1}, \ref{fig:infTNIPredict11}, \ref{fig:infTNIPredict1} and \ref{fig:siru12}. \\
Particularly figure \ref{fig:fitfunccomp13zoom} and figures \ref{fig:fitfunccomp3zoom} shows results of parametric solve of the SIRU model \eqref{sir3} with the parameters $\phi=0$. Hence, we see that the time of the peak becomes $T_{2}$,  the date where stronger measures are taken. We choose as sample $T_{2}=70$ which correspond to 2020, May 10.

\noindent The figures \ref{fig:infPredict11}, \ref{fig:infPredict1}, \ref{fig:infRUPredict11}, \ref{fig:infRUPredict1}, \ref{fig:infTNIPredict11}, \ref{fig:infTNIPredict1} show parametric plot of the infected $I(t)$, the reported $I_{R}(t)$ and unreported $I_{U}(t)$ and total number of infected $\tilde{TNI}(t)$, with varying values of the parameter $\phi$.\\
The figures \ref{fig:siru12} and \ref{fig:siru32} show again, in different range of the ordinate axis, parametric plot of the infected $I(t)$, the reported $I_{R}(t)$ and unreported $I_{U}(t)$ and total number of infected $\tilde{TNI}(t)$, with varying values of the parameter $\phi$.         

\begin{figure}[H]
	\centering
	\subfloat[Parameric plot of the reported $I_{R}(t)$ and unreported $I_{U}(t)$ infected case.\label{fig:infRUPredict0}]{\includegraphics[width=0.45\linewidth]{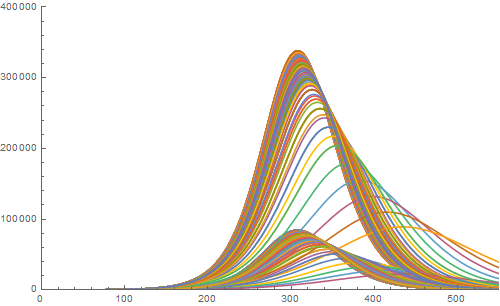}}\qquad
	\subfloat[Zoom of the parameric plot of the reported $I_{R}(t)$ and unreported $I_{U}(t)$ infected case.\label{fig:infRUPredict0zoom1}]{\includegraphics[width=0.45\linewidth]{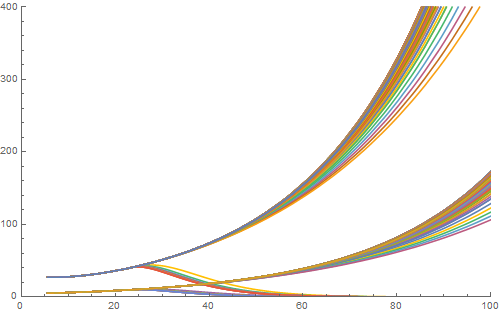}}\qquad
	\subfloat[Parameric plot of the infected $I(t)$.\label{fig:infPredict0}]{\includegraphics[width=0.45\linewidth]{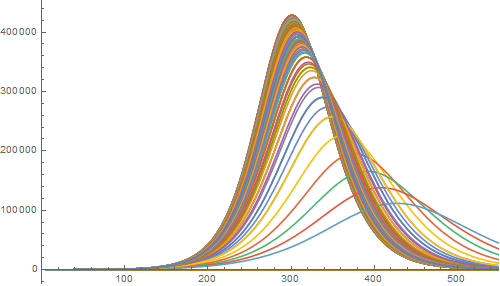}}\qquad
	\subfloat[Zoom of parameric plot of the infected $I(t)$.\label{fig:infPredict0zoom1}]{\includegraphics[width=0.45\linewidth]{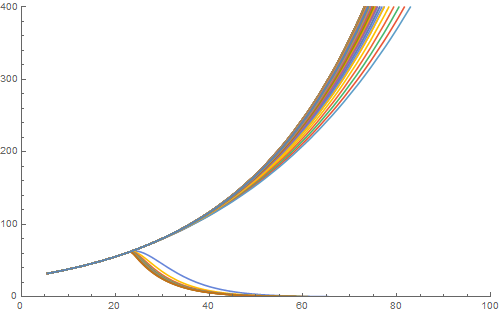}}\qquad
	\subfloat[Parameric plot of the total number of infected case $\tilde{TNI}(t)$.\label{fig:infTNIPredict0}]{\includegraphics[width=0.45\linewidth]{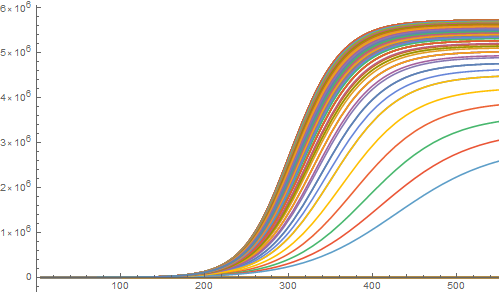}}\qquad
	\subfloat[zoom of the parameric plot of the total number of infected case $\tilde{TNI}(t)$, the reported $I_{R}(t)$ and unreported $I_{U}(t)$ infected case.\label{fig:infTNIRUPredict0}]{\includegraphics[width=0.45\linewidth]{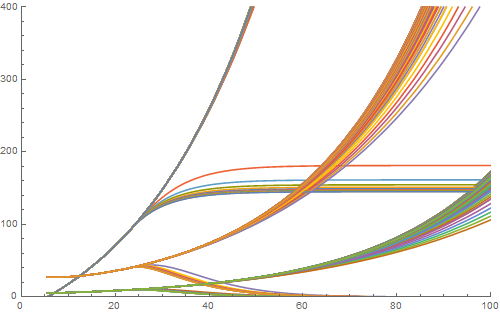}}	
	\caption{View of different sizes parametric plot of the SIRU model\eqref{sir2} using \eqref{beta1},  with measures taken at a time $T$. From top to down corresponding to increasing values of the parameter $\delta/p$, with $\delta\in[0,\ 50]$ by step $5$ and $p\in [1,\ 10001]$ by step of $1000$. On the abscissa axis, the graduation 55 represents 2020, April 25.}
	\label{fig:siru0}
\end{figure}

\begin{figure}[H]
	\centering
	\subfloat[The fit function to total number of infected data $TNI(t)$ in red line, The total number of infected case $\tilde{TNI}(t)$ in yellow, the reported case $I_{R}(t)$ in blue line, the unreported case $I_{U}(t)$ in black line, the first fit function to total number of infected data in green line.\label{fig:fitfunccomp12}]{\includegraphics[width=0.45\linewidth]{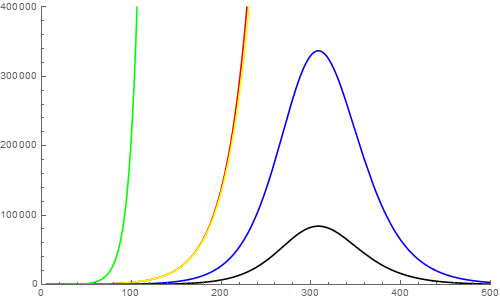}}\qquad
	\subfloat[Plot, with maximal ordinate value fixed to $1000$. The fit function to total number of infected data $TNI(t)$ and the total number of infected case $\tilde{TNI}(t)$ in yellow have the same path.\label{fig:fitfunccomp11}]{\includegraphics[width=0.45\linewidth]{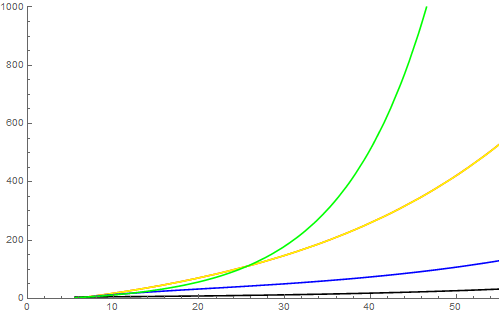}}\qquad
	\subfloat[Plot considering stronger measures taken at a time $T_{2}$ and  $\phi=1$.  \label{fig:fitfunccomp13zoom}]{\includegraphics[width=0.45\linewidth]{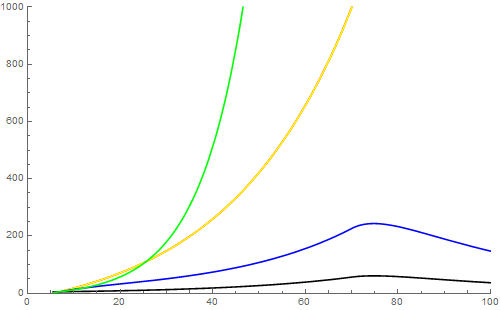}}\qquad
	\subfloat[Parameric plot of the infected case $I(t)$. \label{fig:infPredict11}]{\includegraphics[width=0.45\linewidth]{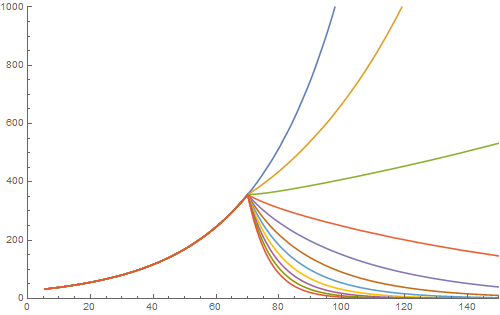}}\qquad
	\subfloat[Parameric plot of the reported $I_{R}(t)$ and unreported $I_{U}(t)$ infected case.\label{fig:infRUPredict11}]{\includegraphics[width=0.45\linewidth]{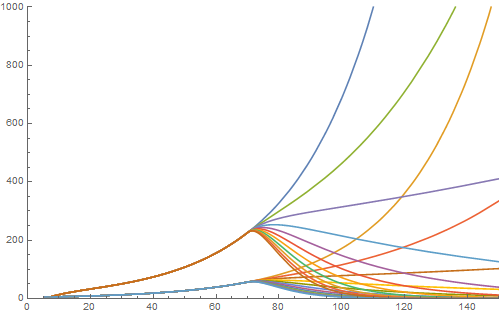}}\qquad
	\subfloat[Parameric plot of the total number infected case $\tilde{TNI}(t)$. \label{fig:infTNIPredict11}]{\includegraphics[width=0.45\linewidth]{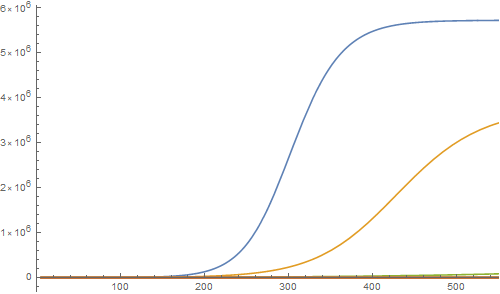}}	
	\caption{Plot and parametric plot of the SIRU model\eqref{sir2}, using \eqref{beta1} with $\delta=1$ and $p=5000$ and \eqref{beta4} with $\phi$ as parameter, fitted to data of Senegal. With stronger measures taken at a time $T_{2}$.  From top to down corresponding to increasing values of the parameter $\phi$ from $0$ to $1$ by step of $0.1$. On the abscissa axis, the graduation 55 represents 2020, April 25.}
	\label{fig:siru11}
\end{figure}

\begin{figure}[H]
	\centering
	\subfloat[Parameric plot, with as maximum value on the ordinate axis fixed at $400000$, of the reported $I_{R}(t)$ and unreported $I_{U}(t)$ infected case.\label{fig:infRUPredict11zoom1}]{\includegraphics[width=0.45\linewidth]{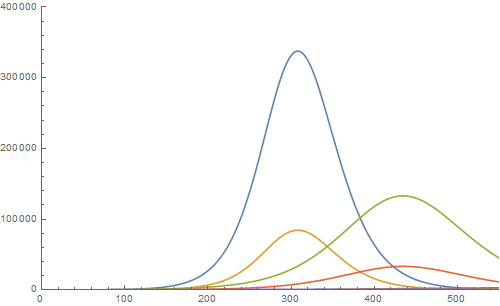}}\qquad
	\subfloat[Parameric plot, with as maximum value on the ordinate axis fixed at $100000$, of the reported $I_{R}(t)$ and unreported $I_{U}(t)$ infected case.\label{fig:infRUPredict11zoom2}]{\includegraphics[width=0.45\linewidth]{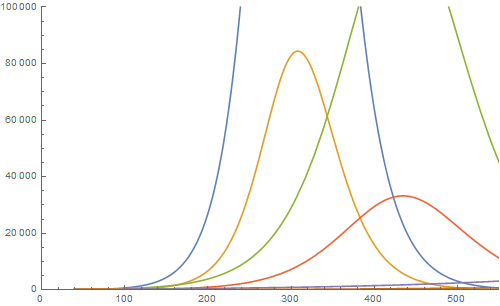}}\qquad
	\subfloat[Parameric plot, with as maximum value on the ordinate axis fixed at $40000$, of the reported $I_{R}(t)$ and unreported $I_{U}(t)$ infected case.\label{fig:infRUPredict11zoom3}]{\includegraphics[width=0.45\linewidth]{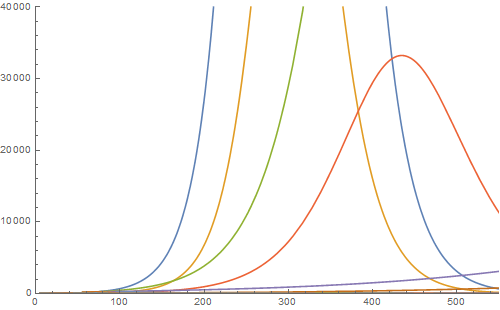}}\qquad
	\subfloat[Parameric plot, with as maximum value on the ordinate axis fixed at $40000$, of the reported $I_{R}(t)$ and unreported $I_{U}(t)$ infected case.\label{fig:infRUPredict11zoom4}]{\includegraphics[width=0.45\linewidth]{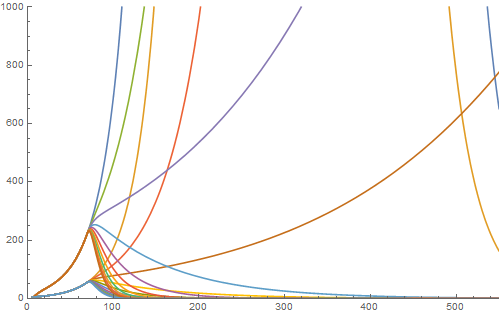}}\qquad
	\subfloat[Parameric plot, with maximal size $10000$, of  $\tilde{TNI}(t)$.\label{fig:infTNIPredict11zoom}]{\includegraphics[width=0.45\linewidth]{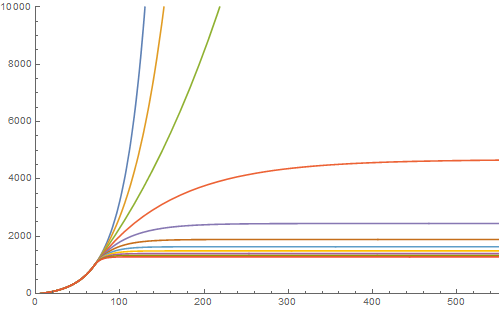}}\qquad
	\subfloat[Parameric plot, with maximal size $5000$, of  $\tilde{TNI}(t)$,  $I_{R}(t)$ and  $I_{U}(t)$.\label{fig:infTNIRUPredict11}]{\includegraphics[width=0.45\linewidth]{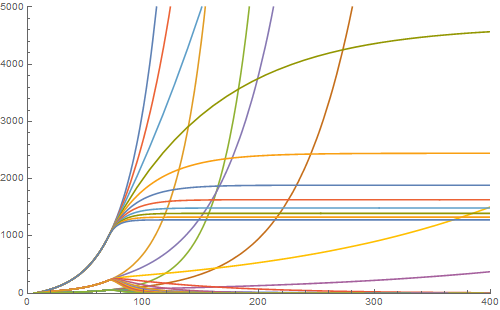}}	
	\caption{View of different sizes parametric plot of the SIRU model\eqref{sir2} using \eqref{beta3},  with stronger measures taken at a time $T_{2}$. From top to down corresponding to increasing values of the parameter $\phi$ from $0$ to $1$ by step of $0.1$. On the abscissa axis, the graduation 55 represents 2020, April 25.}
	\label{fig:siru12}
\end{figure}
\begin{figure}[H]
	\centering
	\subfloat[The fit function to total number of infected data $TNI(t)$ in red line, The total number of infected case $\tilde{TNI}(t)$ in yellow, the reported case $I_{R}(t)$ in blue line, the unreported case $I_{U}(t)$ in black line, the first fit function to total number of infected data in green line.\label{fig:fitfunccomp2}]{\includegraphics[width=0.45\linewidth]{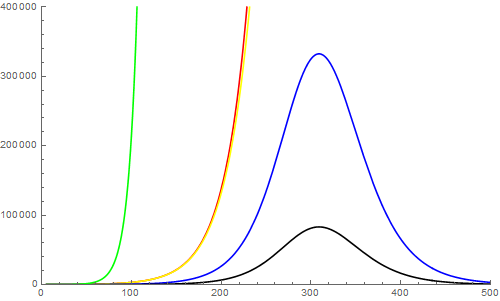}}\qquad
	\subfloat[Plot, with maximal value fixed to $1000$. The fit function to total number of infected data $TNI(t)$ in red line and the total number of infected case $\tilde{TNI}(t)$ in yellow have the same path. \label{fig:fitfunccomp1}]{\includegraphics[width=0.45\linewidth]{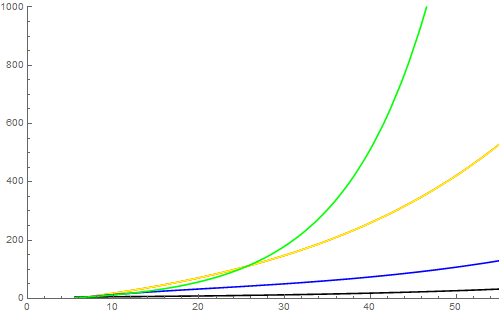}}\qquad
	\subfloat[Plot considering stronger measures taken at a time $T_{2}$ and $\phi=1$.\label{fig:fitfunccomp3zoom}]{\includegraphics[width=0.45\linewidth]{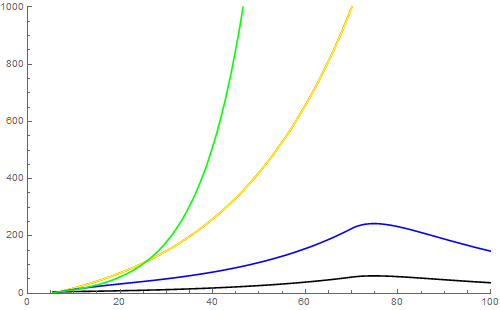}}\qquad
	\subfloat[Parameric plot of the infected case $I(t)$. \label{fig:infPredict1}]{\includegraphics[width=0.45\linewidth]{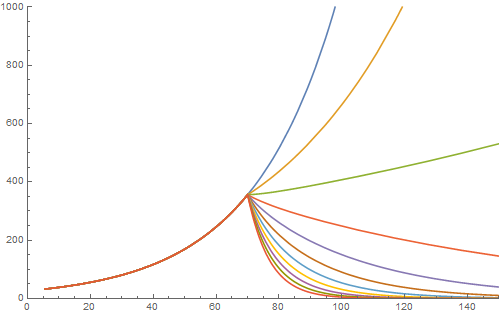}}\qquad
	\subfloat[Parameric plot of the reported $I_{R}(t)$ and unreported $I_{U}(t)$ infected case.\label{fig:infRUPredict1}]{\includegraphics[width=0.45\linewidth]{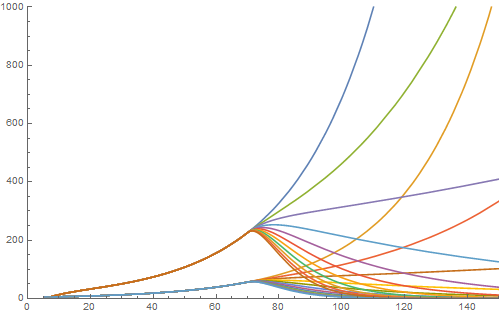}}\qquad
	\subfloat[Parameric plot of the total number infected case $\tilde{TNI}(t)$.\label{fig:infTNIPredict1}]{\includegraphics[width=0.45\linewidth]{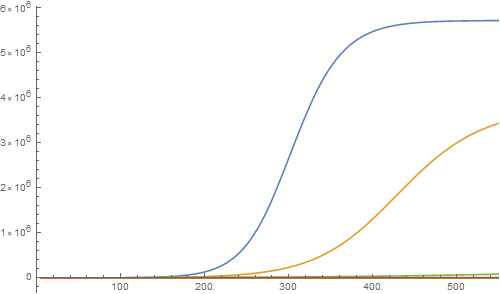}}	
	\caption{Plot and parametric plot of the SIRU model\eqref{sir2},  using \eqref{beta3} with $\varphi=10^{-5}$ and \eqref{beta4} with $\phi$ as parameter, fitted to data of Senegal. With stronger measures taken at a time $T_{2}$.  From top to down coresponding to increasing values of the parameter $\phi$ from $0$ to $1$ by step of $0.1$. On the abscissa axis, the graduation 55 represents 2020, April 25.}
	\label{fig:siru31}
\end{figure}
\begin{figure}[H]
	\centering
	\subfloat[Parameric plot, with as maximum value on the ordinate axis fixed at $400000$, of the reported $I_{R}(t)$ and unreported $I_{U}(t)$ infected case.\label{fig:infRUPredict1zoom1}]{\includegraphics[width=0.45\linewidth]{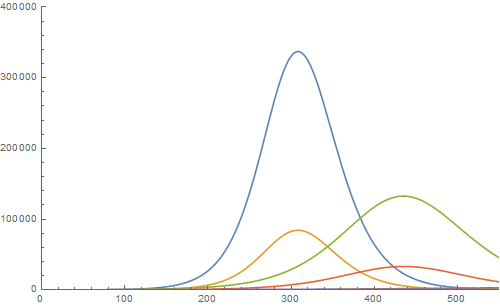}}\qquad
	\subfloat[Parameric plot, with as maximum value on the ordinate axis fixed at $100000$, of the reported $I_{R}(t)$ and unreported $I_{U}(t)$ infected case.\label{fig:infRUPredict1zoom2}]{\includegraphics[width=0.45\linewidth]{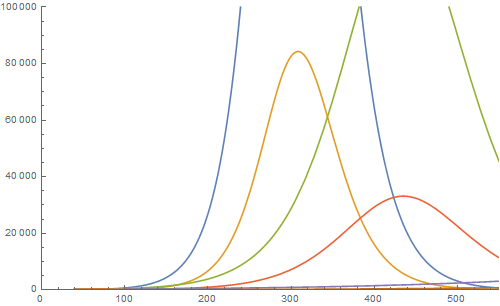}}\qquad
	\subfloat[Parameric plot, with as maximum value on the ordinate axis fixed at $40000$, of the reported $I_{R}(t)$ and unreported $I_{U}(t)$ infected case.\label{fig:infRUPredict1zoom3}]{\includegraphics[width=0.45\linewidth]{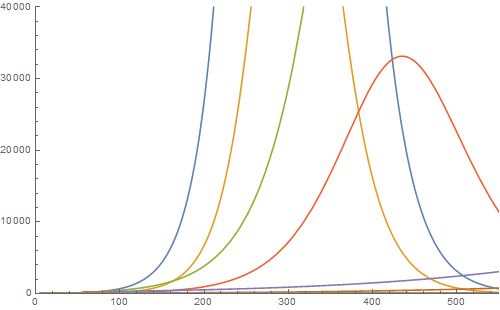}}\qquad
	\subfloat[Parameric plot, with as maximum value on the ordinate axis fixed at $40000$, of the reported $I_{R}(t)$ and unreported $I_{U}(t)$ infected case.\label{fig:infRUPredict1zoom4}]{\includegraphics[width=0.45\linewidth]{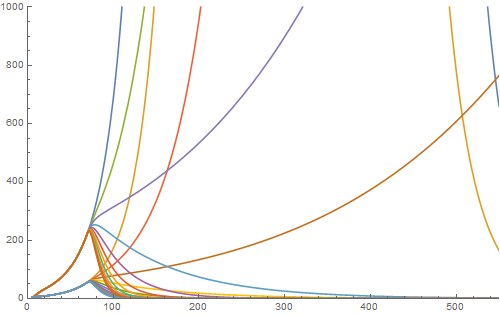}}\qquad
	\subfloat[Parameric plot, with maximal size $10000$, of  $\tilde{TNI}(t)$.\label{fig:infTNIPredict1zoom}]{\includegraphics[width=0.45\linewidth]{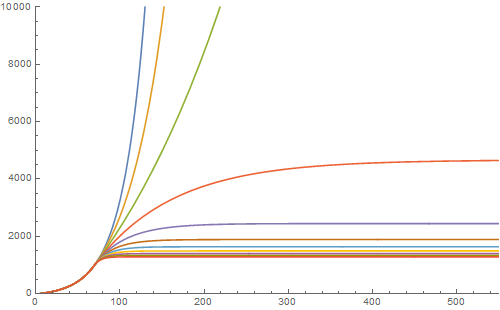}}\qquad
	\subfloat[Parameric plot, with maximal size $5000$, of $\tilde{TNI}(t)$,  $I_{R}(t)$ and  $I_{U}(t)$.\label{fig:infTNIRUPredict1}]{\includegraphics[width=0.45\linewidth]{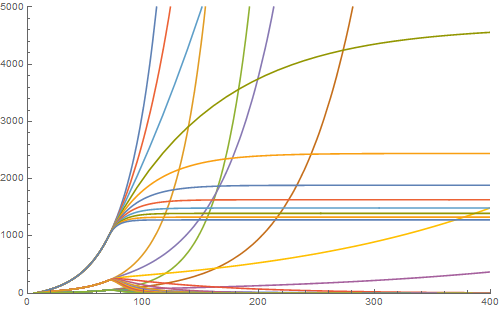}}	
	\caption{View of different sizes parametric plot of the SIRU model\eqref{sir2}  using \eqref{beta4}, with stronger measures taken at a time $T_{2}$. From top to down corresponding to increasing values of the parameter $\phi$ from $0$ to $1$ by step of $0.1$. On the abscissa axis, the graduation 55 represents 2020, April 25.}
	\label{fig:siru32}
\end{figure}
\subsection{Comparative forecasting with machine learning}
\label{sec:simul2}
\noindent The forecasting is done with two data set. For both data from March 02, to April 25, 2020 and March 02, to May 12, 2020 we carry out simulations for a longer time and forecast the potential trends of the COVID-19 pandemic in Senegal. The predicted cumulative number of confirmed cases are first plotted for periods until May 26, June 10 and Sept. 18, 2020 ahead forecast with Prophet, with 95\% prediction intervals. 

\noindent The confirmed predictions for Senegal, using Prophet are given in Figure \ref{fig:mlearn2} (see Tables \ref{data2_May26}, \ref{data2_June10} and \ref{data2_Sept18} for the value of the confidence interval).\\

\noindent The figures \ref{fig:mlearn1} shows forecasting using Neural Networks method of Predict. Particularly the figure \ref{fig:mlearn1} shows two forecasting, one based only on data and an other obtained by training the Neural Networks method with data and contact rate. The prediction are done until 2020, May 26, June 10  and September 18.\\

\begin{table}[H]
	\footnotesize  
	\caption{\small{Predicted cumulative confirmed cases using SIRU model. Forecasting until the dates May 26, June 10 and September 18,  2020.}}\label{tabfitforecast}
	\begin{center}
		\begin{tabular}{|c|c|c|c|c|c|} 
			\hline
			Until date $J$ &  $J-4$        &  $ J-3$       & $J-2$         & $ J-1$         &$J$           \\ \hline
			$J=$2020-05-26 & 1552.118087 & 1614.691394 & 1679.634752 & 1747.037931  &1816.9941   \\ \hline
			$J=$2020-06-10 & 2785.471575 & 2894.759872 & 3008.187615 & 3125.911594  &3248.094533 \\ \hline
			$J=$2020-09-18 & 118687.0026 & 123186.2236 & 127855.8589 & 132702.3632  &137732.4357 \\ \hline
		\end{tabular}
	\end{center}
\end{table} 
\begin{figure}[H]
	\centering
	\subfloat[Plot until May 26,2020. \label{fig:fitfuncpred85} ]{\includegraphics[width=0.45\linewidth]{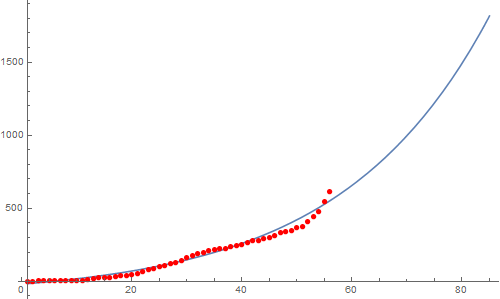}}\qquad
	\subfloat[Plot Plot until June 10,2020.\label{fig:fitfuncpred100}]{\includegraphics[width=0.45\linewidth]{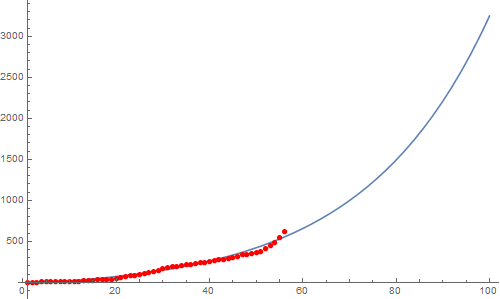}} \qquad
	\subfloat[Plot until September 18, 2020.\label{fig:fitfuncpred200}]{\includegraphics[width=0.45\linewidth]{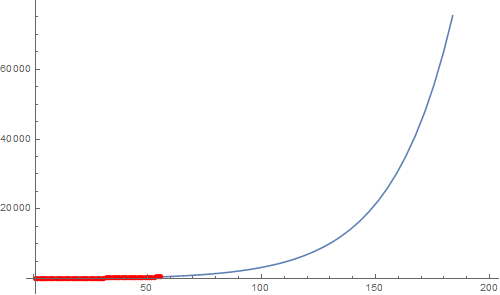}}
	\caption{Plot the function $TNI(t)$ in blue line with data set 1 in red dotted, until the dates 2020, May 26, June 10 and September 18 corresponding to the graduations $85$, $100$ and $200$ on the abscissa axis.}
	\label{fig:fitforecast}
\end{figure}		
\begin{table}[H]
	\footnotesize 
	\caption{\small{Data set 2. Predicted cumulative confirmed cases $\sim$May 26, 2020, from top to down using Prophet and Neural Networks method of Predict. On the down, from left to right with data set \& contact rate and data set only.}}\label{data2_May26}
	\begin{center}
		\begin{tabular}{|c|c|c|c|} 
			\hline
			Date  &         $ \hat{y}$ &    $\hat{y}_{lower}$ &    $\hat{y}_{upper}$   \\
			\hline
			2020-05-22& 	2695.292384 &	2598.664351& 	2807.280554\\ \hline
			2020-05-23& 	2777.228597 &	2663.535151 &	2897.959707\\ \hline
			2020-05-24& 	2853.863973 &	2718.853527 &	2994.394743\\ \hline
			2020-05-25& 	2945.018292 &	2802.374207 &	3095.868137\\ \hline
			2020-05-26& 	3024.913096 &	2857.548751 &	3196.443359\\ \hline
		\end{tabular}\\
		\vspace{0.02cm}	
		\begin{tabular}{|c|c|c|c|c|c|c|} 
			\hline
			Date  &         $ P^{\beta}$ &    $P^{\beta}_{lower}$ &    $P^{\beta}_{upper}$  &         $ P$ &    $P_{lower}$ &    $P_{upper}$  \\\hline
			2020-05-22& 2509.179622 & 2438.869281 & 2579.489963 & 2714.814079 & 2647.304599 & 2782.323560\\ \hline
			2020-05-23& 2575.632714 & 2505.322373 & 2645.943055 & 2803.067798 & 2735.558317 & 2870.577278\\ \hline
			2020-05-24& 2641.331365 & 2571.021025 & 2711.641706 & 2891.331968 & 2823.822488 & 2958.841449\\ \hline
			2020-05-25& 2706.266493 & 2635.956152 & 2776.576834 & 2979.590415 & 2912.080934 & 3047.099895\\ \hline
			2020-05-26& 2770.436355 & 2700.126014 & 2840.746696 & 3067.828454 & 3000.318974 & 3135.337935\\ \hline
		\end{tabular}
	\end{center}
\end{table} 
\begin{table}[H]
	\footnotesize
	\caption{\small{Data set 2. Predicted cumulative confirmed cases $\sim$June 10, 2020, from top to down using Prophet and Neural Networks method of Predict. On the down, from left to right with data set \& contact rate and data set only.}}\label{data2_June10} 
	\begin{center}
		\begin{tabular}{|c|c|c|c|} 
			\hline
			Date  &         $ \hat{y}$ &    $\hat{y}_{lower}$ &    $\hat{y}_{upper}$   \\
			\hline
			2020-06-06 &	3887.495614& 	3491.259628& 	4284.386057\\ \hline
			2020-06-07& 	3964.130990 &	3541.984441 &	4387.085352\\ \hline
			2020-06-08& 	4055.285309 &	3616.658728 &	4496.788853\\ \hline
			2020-06-09& 	4135.180113 &	3685.960670 &	4604.522925\\ \hline
			2020-06-10 &	4206.725884 &	3726.064523 &	4685.455882\\ \hline
		\end{tabular}\\
		\vspace{0.02cm}
		\begin{tabular}{|c|c|c|c|c|c|c|} 
			\hline
			Date  &         $ P^{\beta}$ &    $P^{\beta}_{lower}$ &    $P^{\beta}_{upper}$  &         $ P$ &    $P_{lower}$ &    $P_{upper}$  \\
			\hline
			2020-06-06 &3409.698565 & 3339.388223 & 3480.008905& 4034.675309 & 3967.165829 &  4102.184790\\ \hline
			2020-06-07 &3462.752283 & 3392.441942 & 3533.062624& 4122.100051 & 4054.590571 &  4189.609532\\ \hline
			2020-06-08 &3515.229377 & 3444.919036 & 3585.539718& 4209.437690 & 4141.928209 &  4276.947171\\ \hline
			2020-06-09 &3567.143036 & 3496.832695 & 3637.453377& 4296.685488 & 4229.176007 &  4364.194968\\ \hline
			2020-06-10 &3618.508442 & 3548.198100 & 3688.818782& 4383.845435 & 4316.335955 &  4451.354916\\ \hline
		\end{tabular}
	\end{center}
\end{table} 

\begin{table}[H]
	\footnotesize
	\caption{\small{Data set 2. Predicted cumulative confirmed cases $\sim$Sept 18, 2020, from top to down using Prophet and Neural Networks method of Predict. On the down, from left to right with data set \& contact rate and data set only.}}\label{data2_Sept18} 
	\begin{center}
		\begin{tabular}{|c|c|c|c|} 
			\hline
			Date  &         $ \hat{y}$ &    $\hat{y}_{lower}$ &    $\hat{y}_{upper}$   \\
			\hline
			2020-09-14& 	11827.154428 &	7680.757057 &	15751.856303\\ \hline
			2020-09-15& 	11907.049231 	&7721.053577 &	15888.511751\\ \hline
			2020-09-16& 	11978.595002 	&7769.517365 &	15984.853526\\ \hline
			2020-09-17& 	12051.967485 	&7776.250955 &	16102.077092\\ \hline
			2020-09-18& 	12132.562027 	&7824.092604 &	16225.629593\\ \hline
		\end{tabular}\\
		\vspace{0.02cm}
		\begin{tabular}{|c|c|c|c|c|c|c|} 
			\hline
			Date  &         $ P^{\beta}$ &    $P^{\beta}_{lower}$ &    $P^{\beta}_{upper}$  &         $ P$ &    $P_{lower}$ &    $P_{upper}$  \\
			\hline
			2020-09-14&7404.458859 & 7334.148519 & 7474.769201  & 12499.239975 & 12431.730495 & 12566.749456\\ \hline
			2020-09-15&7439.797975 & 7369.487634 & 7510.108316  & 12582.463095 & 12514.953615 & 12649.972576\\ \hline
			2020-09-16&7475.119671 & 7404.809330 & 7545.430012  & 12665.674270 & 12598.164790 & 12733.183751\\ \hline
			2020-09-17&7510.423945 & 7440.113604 & 7580.734286  & 12748.868522 & 12681.359041 & 12816.378002\\ \hline
			2020-09-18&7545.708310 & 7475.397969 & 7616.018651  & 12832.053814 & 12764.544334 & 12899.563295\\ \hline
		\end{tabular}
	\end{center}
\end{table} 
\begin{figure}[H]
	\centering
	\subfloat[The forecasting is until 2020, May 26 which correspond to the graduation $85$ on the abscissa axis.\label{fig:machinelearncompa0ic}]{\includegraphics[width=0.45\linewidth]{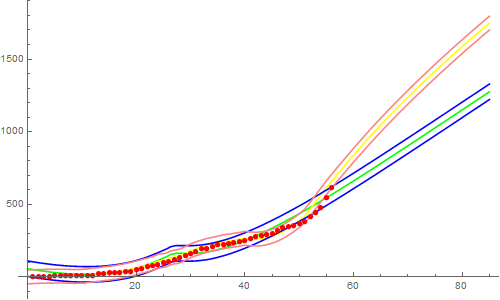}}\qquad	
	\subfloat[The forecasting is until 2020, May 26 which correspond to the graduation $85$ on the abscissa axis.\label{fig:machinelearncompa0ic2}]{\includegraphics[width=0.45\linewidth]{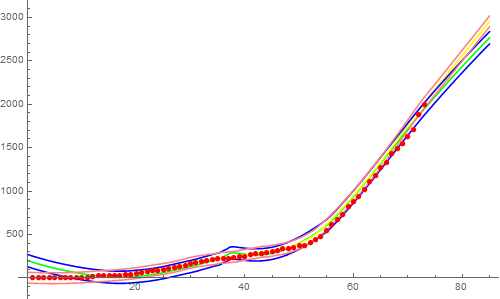}}\qquad	
	\subfloat[The forecasting is until 2020, June 10 which correspond to the graduation $100$ on the abscissa axis.\label{fig:machinelearncompa1-bis2ic}]{\includegraphics[width=0.45\linewidth]{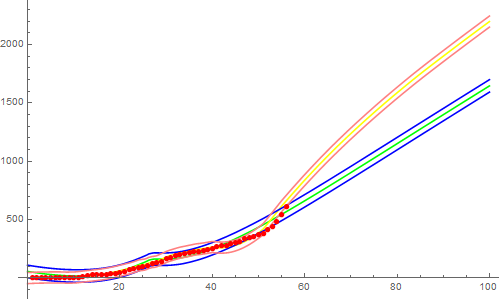}}\qquad
	\subfloat[The forecasting is until 2020, June 10 which correspond to the graduation $100$ on the abscissa axis.\label{fig:machinelearncompa2-bis2ic2}]{\includegraphics[width=0.45\linewidth]{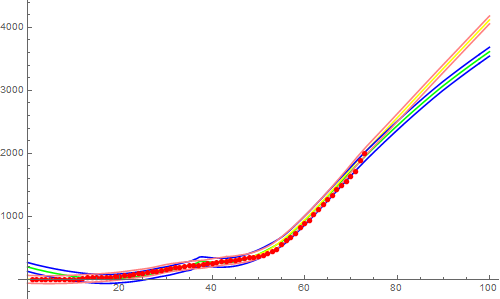}}\qquad
	\subfloat[The forecasting is until 2020, September 18 which correspond to the graduation $200$ on the abscissa axis.\label{fig:machinelearncompa2-bis2ic}]{\includegraphics[width=0.45\linewidth]{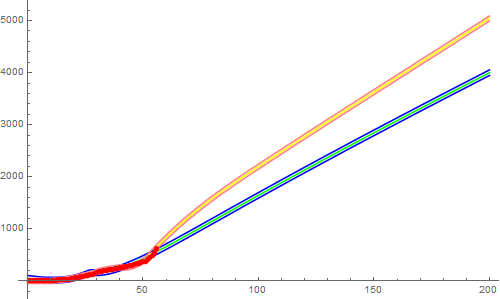}}\qquad
	\subfloat[The forecasting is until 2020, September 18 which correspond to the graduation $200$ on the abscissa axis.\label{fig:machinelearncompa3-bis2ic2}]{\includegraphics[width=0.45\linewidth]{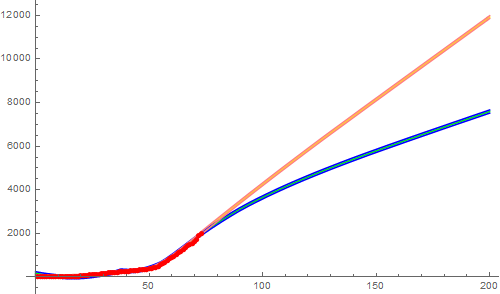}}\qquad
	\caption{Comparative forecasting, using Neural Networks, of the total number of infected cases. On the abscissa axis the graduations $85$, $100$ and $200$  correspond to the dates 2020, May 26, June 10 and September 18. Forecasting using cumulative data only in yellow curve with confidence interval in orange, using both cumulative and contact rate data in green curve with confidence interval in blue, and data in red dotted. In the left plot using data set 1 and in the right plot using data set 2.}
	\label{fig:mlearn1}
\end{figure}
\begin{figure}[H]
	\centering
	\subfloat[The forecasting is until the date 2020, May 26.\label{fig:data1_May26}]{\includegraphics[width=0.45\linewidth]{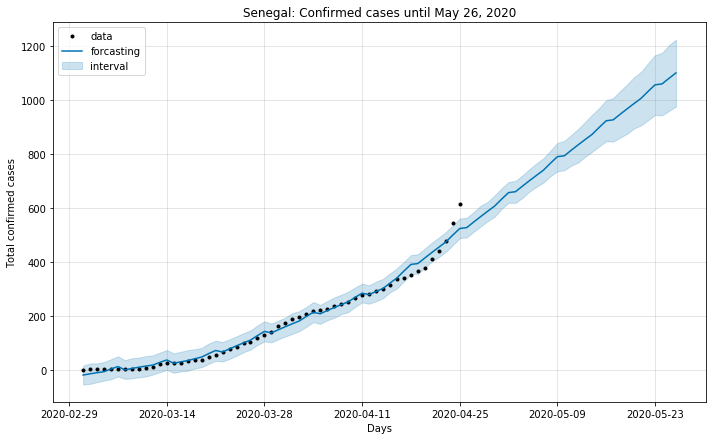}}\qquad
	\subfloat[The forecasting is until the date 2020, May 26.\label{fig:data2_May26}]{\includegraphics[width=0.45\linewidth]{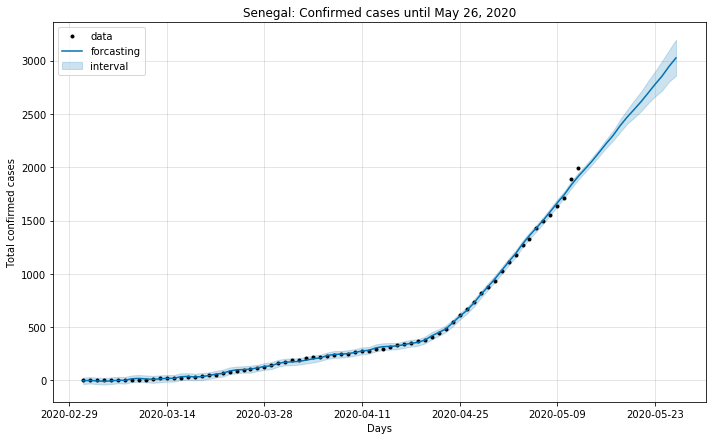}}\qquad	
	\subfloat[The forecasting is until the date 2020, June 10.\label{fig:data1_June10}]{\includegraphics[width=0.45\linewidth]{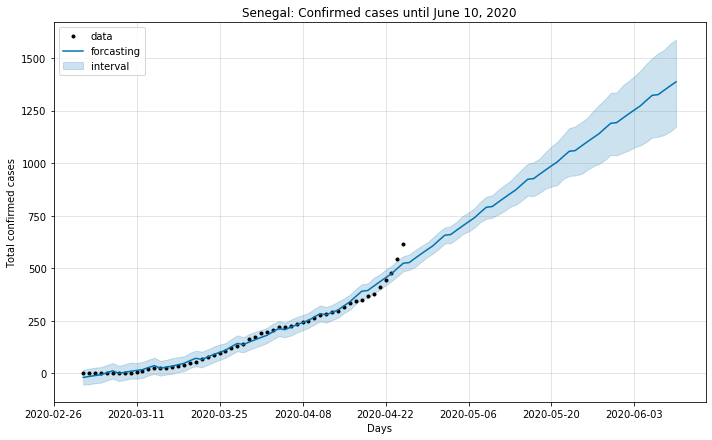}}\qquad
	\subfloat[The forecasting is until the date 2020, June 10.\label{fig:data2_June10}]{\includegraphics[width=0.45\linewidth]{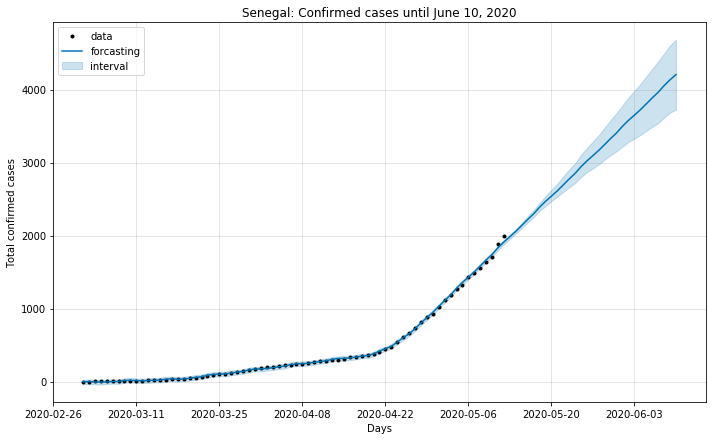}}\qquad
	\subfloat[The forecasting is until the date 2020, September 18.\label{fig:data1_Sept18}]{\includegraphics[width=0.45\linewidth]{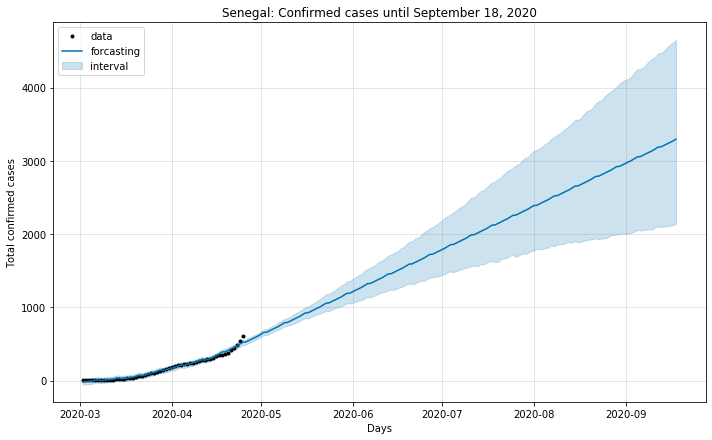}}\qquad
	\subfloat[The forecasting is until the date 2020, September 18.\label{fig:data2_Sept18}]{\includegraphics[width=0.45\linewidth]{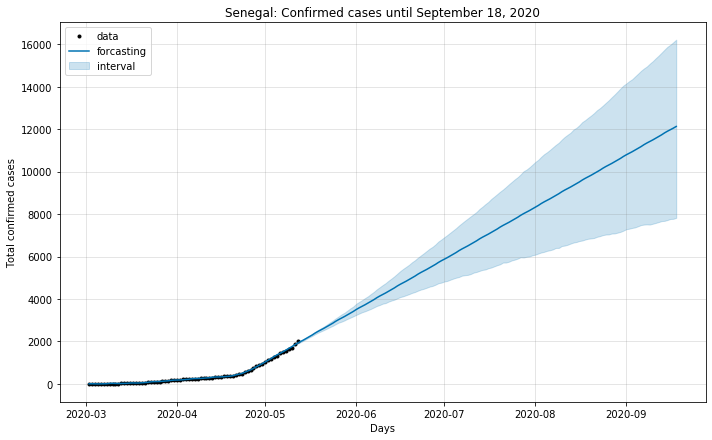}}\qquad
	\caption{Forecasting, using Prophet, of the total number of infected cases until the dates 2020, May 26, June 10 and September 18, from top to down. In left plot using data set 1 and in the right plot using data set 2.}
	\label{fig:mlearn2}
\end{figure}

\section{Discussion}
\label{sec:discuss}

Analysis of the new trend in the data from March 2 to April 24, 2020 shows a change in the trajectory of the total number of cases. That causes a reduction of the maximum value of the peak compared to what it would have been without the measures taken on March 23, 2020. See figures \ref{fig:fitfunccomp12}, \ref{fig:fitfunccomp11}, \ref{fig:fitfunccomp1}, \ref{fig:fitfunccomp2}. \\
By considering new nationwide measures of Senegal, which we have chosen in this study to fix on the date of May 10, 2020, we note that the maximum value of the peak decreases according to the force of the measures taken. Likewise, the time of the peak is postponed as shown by the parametric plots in figures \ref{fig:siru11}, \ref{fig:siru12}, \ref{fig:siru31} and \ref{fig:siru32}. \\
Since the first measures on March 23, 2020, Senegal has laughed at additional measures such as the closing of markets, shops, stores and other public places, with an opening calendar. We have therefore chosen May 10, 2020 as a date from which the additional measures can begin to take effect in reducing contamination. \\
We see that depending on the strength of these measures, the evolution of the disease can lose its exponential nature or become slower. \\
The prediction with neural networks and Prophet show an optimistic situation. The forecasting based only on data and those on contact rates show a slow evolution as shown in figures \ref{fig:mlearn1} and \ref{fig:mlearn2}.
We see that the curve obtained using the contact rate function in training of the neural networks, goes below that obtained only using the data on the total number of cases.

\section{Materials and methods}
\label{sec:app}
\subsection{Basic proprieties}
Let's set $E=\{(S,I,I_{R},I_{U})\in \mathbb{R}^{4}_{+}:\ S+I+I_{R}+I_{U}\leq N \}$, $N$ is the initial population. 
Replacing an initial solution $(S,I,I_{R},I_{U})\in E$ in the system \eqref{sir1}, we obtain:
\begin{equation}
\left\{ \begin{array}{l}
\displaystyle \dot{S}\leq 0\\
\displaystyle \dot{I}\geq - \nu I(t)\\
\displaystyle \dot{I_{R}}\geq -\eta I_{R}(t)\\
\displaystyle \dot{I_{U}}\geq -\eta I_{U}(t)
\end{array}
\right.
\end{equation}
Solving these differential inequalities gives $I(t)\geq K \exp(-\nu t)$, $I_{R}(t)\geq K \exp(-\eta t)$ and $I_{U}\geq K \exp(-\eta t)$. Hence $I,\ I_{R}, \ I_{U}\geq 0 $. \\
By considering an initial solution with $S=0$, we get $\dot{S}=0$. Then $S(t)\geq 0$.\\
Now, summing the equations of the system \eqref{sir1}, we obtain $\dot{S}+\dot{I}+\dot{I_{R}}+\dot{I_{U}}\leq -\eta (I_{R}+I_{U})\leq 0$.   \\
Then $S+I+I_{R}+I_{U}\leq S_{0}+I_{0}+I_{R0}+I_{U0}\leq N$. Finally $(S,I,I_{R},I_{U})$ remain in the set $E$. 
This implies that $E$ is a positively-invariant set under the flow described by \eqref{sir1}. In addition, the model can be considered as being epidemiologically and mathematically well-posed.

\noindent Let's set 
\begin{equation}
X=\left(\begin{array}{l} S\\
I\\
I_{R}\\
I_{U}
\end{array}\right), \  f(X(t))=\left(\begin{array}{l} 
-\beta S(t) (I(t)+I_{U}(t))\\
\beta S(t)( I(t)+I_{U}(t)) - \nu I(t)\\
\gamma\nu I(t)-\eta I_{R}(t)\\
(1-\gamma)\nu I(t)-\eta I_{U}(t)
\end{array}\right)\ \textrm{and}\ X_{0}=\left(\begin{array}{l} S_{0}\\
I_{0}\\
I_{R0}\\
I_{U0}
\end{array}\right)\in E
\end{equation}
Then the system \eqref{sir1} with initial condition can be rewritten in the following form:
\begin{equation}
\left\{\begin{array}{l}
\dot{X}=f(X(t))\\
X(t_{0})=X_{0}
\end{array}\right.
\label{gfsir}
\end{equation}
\begin{theo}
	The system \eqref{gfsir} has a unique solution $X(t)$ for $t\geq t_{0}$. 
	\label{th1}
\end{theo}
\begin{rem}
	The proof of the above theorem come from the Cauchy-Lipschitz theorem and the Fundamental Existence-Uniqueness theorem in \cite{Perko:1991}, since $f\in C^{1}(E)$.	
\end{rem}
\noindent Let's set the function $g:\mathbb{R}_{+}\times \mathbb{R}\rightarrow \mathbb{R}$ such that $g(\beta,X)=f(X)$.
\begin{equation*}
\begin{array}{rl}
g:\mathbb{R}_{+}\times \mathbb{R}&\rightarrow \mathbb{R}\\
(\beta,X)&\mapsto f(X)
\end{array},\ \tilde{f}(X)=\left\{\begin{array}{ll}
\displaystyle g(\beta,X)& \textrm{if}\ t\in[t_{0}, T]\\
\displaystyle g(\beta(\frac{T}{t})^{\delta/p},X)& \textrm{if}\ t> T 
\end{array}
\right.
\end{equation*}   
Then the system \eqref{sir3} can be written in the following form:
\begin{equation}
\left\{\begin{array}{l}
\dot{X}=\tilde{f}(X(t))\\
X(t_{0})=X_{0}
\end{array}\right.
\label{gfsir2}
\end{equation}
\begin{theo}
	For any fixed parameters $\delta$ and $p$ for $\tilde{\beta}$ \eqref{beta1} or $\varphi$ for $\tilde{\beta}$ \eqref{beta3}, the system \eqref{gfsir2} with $\tilde{\beta}$ \eqref{beta1} or $\tilde{\beta}$ \eqref{beta3}, has a unique solution $X(t)$ for $t\geq t_{0}$. 
	\label{th2}
\end{theo}
\begin{proof}[Proof of Theorem \ref{th2}]
	The proof of the existence uniqueness of the system \eqref{gfsir2} with $\tilde{\beta}$ \eqref{beta1} and $\tilde{\beta}$ \eqref{beta3} are similar. Then, we do this for one.  \\	
	The function $\tilde{f}$ is piece-wise continuously differentiable in $\mathbb{R}_{+}\times E$. Then by using the theorem \ref{th1}, we have the existence of a unique solution $X_{1}(t)$ on $[t_{0}, T]$ of the system $\left\{\begin{array}{l}
	\displaystyle	\dot{X}=\displaystyle g(\beta,X(t))\\
	X(t_{0})=X_{0}
	\end{array}\right.$ 
	and for the parameters $\delta$ and $p$ fixed, there is a unique solution $X_{2}(t)$ for $t\geq T$ of the system 
	$\left\{\begin{array}{l}
	\displaystyle\dot{X}=\displaystyle g(\beta(\frac{T}{t})^{\delta/p},X(t))\\
	X(T)=X_{1}(T)
	\end{array}\right.$.\\
	Now defining the function $X(t)=\left\{\begin{array}{ll}
	X_{1}(t)& \textrm{if}\ t\in[t_{0}, T]\\
	X_{2}(t)& \textrm{if}\ t> T 
	\end{array}
	\right.$, we deduce that $X(t)$ is a unique solution continuous in time of the system \eqref{gfsir2}.  
\end{proof}
\begin{rem}
	We can do the same work for the system \eqref{sir3} with $\tilde{\beta}$ replaced by $\tilde{\beta}_{2}$ \eqref{beta4}. Then for any fixed parameters $\delta$, $p$ and $\phi$, the system has a unique solution $X(t)$ for $t\geq t_{0}$. 
\end{rem}
\subsection{Disease Free Equilibrium}
The unique equilibrium of the model \eqref{sir1} is the canonical DFE given by $(S_{0},0,0,0)$.  
\subsection{Parameters estimation}
The estimation of the parameters of the model \eqref{sir1} is done by using techniques in \cite{LMSW:2020}, \cite{Balde:2020}. We fit the cumulative data with an exponential function $TNI(t)=b\exp(c t)-a$. In addition, we assume that the cumulative function can be given in integral form as $TNI(t)=\displaystyle\gamma\nu\int_{t_{0}}^{t}I(s)ds$.\\ Then $TNI(t_{0})=TNI_{0}=b\exp(c t_{0})-a=0$. Thus, we obtain $\displaystyle t_{0}=\displaystyle\frac{\ln(a)-\ln(b)}{c}$.\\
Also we have:
\begin{equation}
I(t)=\dot{TNI}(t)=bc\exp(c t).
\end{equation} 
Then $\displaystyle I(t_{0})=\frac{bc}{\gamma\nu}\exp(c t_{0})=\frac{ac}{\gamma\nu}=I_{0}$ and $\displaystyle\frac{I(t)}{I(t_{0})}=\exp(c(t-t_{0}))$. 
Hence, we obtain
\begin{equation}
I(t)=I(t_{0})\exp(c(t-t_{0})),
\end{equation}
then $\dot{I}(t)=cI(t)$ and $\dot{I}(t_{0})=cI(t_{0})$.\\ 
Let's set $\delta_{1}$ and $\delta_{2}$ such that $I_{1}=\delta_{1} I$ and $I_{2}=\delta_{2} I$. 
Then replacing in the second an third equation of the following system: 
\begin{equation}
\label{modelinf}
\left\{\begin{array}{lcl}
\dot I &=&\beta \,S\,(I+I_{U})-\nu\,I    \\
\dot I_R&=& \gamma\nu \,I - \eta\,I_R\\
\dot I_U &=&  (1-\gamma)\nu \,I - \eta\,I_U, \\
\end{array}\right.
\end{equation}
we obtain 
\begin{align}
\delta_{1}=\frac{\gamma\nu}{c+\eta}=\frac{I_{R0}}{I_{0}}&\\
\delta_{2}=\frac{(1-\gamma)\nu}{c+\eta}=\frac{I_{U0}}{I_{0}}&
\label{delta}
\end{align}
Then introducing \eqref{delta} in the first equation of \eqref{modelinf}, we obtain:
\begin{equation*}
c+\nu=\beta \,S_{0}\,(1+\delta_{2}).
\end{equation*}  
Hence 
\begin{align}
\label{beta}
\beta=\frac{c+\nu}{S_{0}\,(1+\delta_{2})}.
\end{align}
Replacing \eqref{delta} in \eqref{beta}, we obtain:
\begin{align}
\beta=\frac{(c+\nu)(c+\eta)}{S_{0}\,(c+\eta+(1-\gamma)\nu)}.
\end{align}
\subsection{The basic reproductive number $\mathcal{R}_{0}$}
To compute the $\mathcal{R}_{0}$, we use \cite{DriesscheWatmough:2002}. We consider the linearized infected equations at the Disease Free Equilibrium (DFE) $(S_{0},0,0,0,0,0)$, in the system \eqref{sir2}:
\begin{equation}
\left\{ \begin{array}{l}
\displaystyle \frac{dI}{dt}=\beta S_{0}( I(t)+I_{U}(t)) - \nu I(t)\\\\
\displaystyle \frac{dI_{U}}{dt}=(1-\gamma)\nu I(t)-\eta I_{U}(t)\\
\end{array}
\right.
\label{sir4}
\end{equation} 
Extracting the matrix:
\begin{align*} 
M=&\left[\begin{array}{cc}
\beta S_{0}-\nu & \beta S_{0}\\
(1-\gamma)\nu & -\eta\\
\end{array}\right]\\\\
=&\left[\begin{array}{cc}
\beta S_{0} & \beta S_{0}\\
(1-\gamma)\nu & 0\\
\end{array}\right]-\left[\begin{array}{cc}
\nu & 0\\
0 & \eta\\
\end{array}\right]
\end{align*}   
Then the next generation is given by:
\begin{align*} 
FV^{-1}=&\left[\begin{array}{cc}
\beta S_{0}/\nu & \beta S_{0}/\eta\\
(1-\gamma) & 0\\
\end{array}\right]
\end{align*}   
Finally the spectral radius $\rho(FV^{-1})$ gives:
\begin{equation} 
\mathcal{R}_{0}=\frac{\beta S_{0}}{\nu}(1+\frac{(1-\gamma)\nu}{\eta}).
\end{equation}

\section{Conclusion and perspectives}
\label{sec:conclusion}
In this paper, we have analyzed the impact of anti-pandemic measures in Senegal. First, we used techniques of fitting function to the data of the total number of cases. The choice of the data fit function is crucial for the results  since it allows the estimation of the parameters of the compartmental model used. In a second work, we use machine learning tools to also predict the future evolution of the pandemic in Senegal. Also, we have integrated the effects of the measures into the prediction.\\
Depending on the measures taken, the pandemic's trajectory may become slower or lose its exponential nature.
It would be interesting, in the following, to use other functions of a slow nature like the logistical function to fit data and thus obtain different results. A stochastic study using Brownian motions applied to the SIRU compartmental model would also be interesting.

\section*{Acknowledgement}
The authors thanks the Non Linear Analysis, Geometry and Applications (NLAGA) Project for supporting this work.

\section*{Conflict of interest}
The author declares no conflict of interest.


\begin{thebibliography}{99}

	\bibitem{Alpaydin:2010}
	Alpaydin, E. (2010) Introduction to Machine Learning 2nd ed, 584. Adaptive Computation and Machine Learning.
	
	\bibitem{AndMay:1992}
	Anderson, R.~M. and May, R.~M. (1992) Infectious Diseases of Humans, 768. Oxford University Press, Oxford.
	
	\bibitem{Balde:2020}
	Baldé, M.A.M.T. (2020) Fitting SIR model to COVID-19 pandemic data and comparative forecasting with machine learning, medRxiv preprint doi: \url{https://doi.org/10.1101/2020.04.26.20081042}.
	
	\bibitem{GoodfellowBengioCourville:2016}
	Goodfellow, I., Bengio Y. and Courville, A. (2016) Deep Learning, MIT Press. \url{http://www.deeplearningbook.org} 
	
	\bibitem{Heaton:2015}
	Heaton, J. (2015) AIFH Volume 3: Deep Learning and Neural Networks, Heaton Research, Inc, 268. Tracy Heaton (2015).
	
	\bibitem{Hethcote:2000}
	Hethcote, H.~W. (2000) The Mathematics of Infectious	Diseases, Society for Industrial and Applied Mathematics, Vol. 42, No. 4, pp. 599–653.
	
	
	\bibitem{LanPed:2016}
	Langtangen, H.~P. and Pedersen, G.~K. (2016) Scaling of Differential Equations, 152. Springer.
	
	\bibitem{LKM:2020}
	Di Lauro, F. , Kissy, I.Z. and Miller, J.C. (2020) The timing of one-shot interventions for epidemic control, 
	medRxiv preprint doi:\url{ https://doi.org/10.1101/2020.03.02.20030007}. 
	
	\bibitem{LMSW:2020}
	Liu, Z., Magal, P., Seydi, O., and Webb, G. (2020) Understanding Unreported Cases in the 2019-Ncov Epidemic Outbreak in Wuhan, China, and the Importance of Major Public Health Interventions, SSRN: \url{https://ssrn.com/abstract=3530969} or \url{http://dx.doi.org/10.2139/ssrn.3530969}.
	
	\bibitem{LMSW2:2020}
	Liu, Z., Magal, P., Seydi, O., and Webb, G. (2020) Predicting the cumulative number of cases for the
	COVID-19 epidemic in China from early data, medRxiv preprint \url{https://doi.org/10.1101/2020.03.11.20034314}.
	
	\bibitem{Ndiayeetal:2020}
	Ndiaye, B.M., Tendeng, L. and Seck, D. (2020) Analysis of the COVID-19 pandemic by SIR model and machine learning technics for forecasting. \url{https://arxiv.org/abs/2004.01574v1}.
	
	\bibitem{Ndiayeetal2:2020}
	Ndiaye, B.M., Tendeng, L. and Seck, D. (2020) Comparative prediction of confirmed cases with COVID-19 pandemic by machine learning, deterministic and stochastic SIR models. \url{https://arxiv.org/abs/2004.13489}.
	
	\bibitem{Newman:2010}
	Newman, M. E. J. (2010) Networks An Introduction, Oxford University Press, 394.
	
	\bibitem{Perko:1991}
	Perko, L. (1991) Differential equations and dynamical systems, Springer-Verlag New-York.
	
	\bibitem{prophet} Prophet: Automatic Forecasting Procedure, avalailable in \url{https://facebook.github.io/prophet/docs/} or  \url{https://github.com/facebook/prophet}.
	%
	\bibitem{python} Python Software Foundation. Python Language Reference, version 2.7.  Available at \url{http://www.python.org}.
	
	\bibitem{DriesscheWatmough:2002}
	Van den Driessche, P. and Watmough, J. (2002) Reproduction numbers and subthreshold endemic equilibria
	for compartmental models of disease transmission, Mathematical Biosciences 180, 29-48.
	
	\bibitem{WR}
	Wolfram Mathematica, \url{https://www.wolfram.com/mathematica/?source=nav}.
	
\end{thebibliography}
\end{document}